\documentclass[lettersize,journal]{IEEEtran}
\usepackage{amsmath,amsfonts}
\usepackage{algorithmic}
\usepackage{algorithm}
\usepackage{array}
\usepackage{color}
\usepackage[caption=false,font=normalsize,labelfont=sf,textfont=sf]{subfig}
\usepackage{textcomp}
\usepackage{stfloats}
\usepackage{url}
\usepackage{verbatim}
\usepackage{graphicx}
\usepackage{cite}
\usepackage[colorlinks=true,linkcolor=red,citecolor=blue,urlcolor=blue,]{hyperref}
\usepackage{cleveref}
\newtheorem{lemma}{Lemma}
\newtheorem{theorem}{Theorem}
\newtheorem{remark}{Remark}
\newtheorem{corollary}{Corollary}

\crefname{figure}{Fig.}{Figs.} 
\begin{document}

\title{Hybrid Pinching-Fluid Antenna Assisted Wireless Communications: Modeling and Performance Analysis}

\author{{\small Xiao Lin, Yizhe Zhao,~\IEEEmembership{Member,~IEEE}, Xiangyang Wang, Halvin Yang, \IEEEmembership{Member,~IEEE}, and Bingxin Zhang,~\IEEEmembership{Member,~IEEE}}
\thanks{Xiao Lin, Yizhe Zhao and Xiangyang Wang are with the School of Information and Communication Engineering, University of Electronic Science and Technology of China, Chengdu 611731, China (e-mail: xiaolin@std.uestc.edu.cn; yzzhao@uestc.edu.cn; w2714997021@163.com). }
\thanks{Halvin Yang is with the Department of Electrical and Electronic Engineering, Imperial College London, London, U.K (e-mail: halvin.yang@imperial.ac.uk).}
\thanks{Bingxin Zhang is with the State Key
	Laboratory of Novel Software Technology, Nanjing University, Nanjing,
	210008, China, Institute of Intelligent Networks and Communications
	(NINE), Collaborative Innovation Center of Novel Software Technology
	and Industrialization, and School of Intelligent Software and Engineering, Nanjing University (Suzhou Campus), Suzhou, 215163, China (email:
	bxzhang@nju.edu.cn).}
}


\maketitle
\begin{abstract}
Reconfigurable-antenna systems have received increasing attention for their
ability to adapt wireless channels. However, existing architectures exhibit
scenario-dependent limitations: fluid antennas provide strong diversity gains
in rich-scattering environments but offer limited benefits under line-of-sight
(LoS)-dominant conditions, while pinching antennas can effectively reduce path
loss by adjusting the radiation point along a waveguide, yet perform poorly in
severe non-LoS (NLoS) scenarios. This letter proposes a hybrid pinching-fluid antenna system (HPFAS), where 
pinching antenna (PA) is employed at the transmitter and a fluid antenna (FA) is used at
the receiver to jointly exploit LoS enhancement and spatial diversity. 
A tractable channel model is developed, and outage probability
expressions are derived for both single-user and multi-user scenarios.
\begingroup
\color{red}
High-SNR asymptotic analysis and numerical comparisons are further
conducted to characterize the roles of the two antenna mechanisms and
provide scheme-selection insights.
Simulation results validate the analysis and
\color{red}
demonstrate the performance benefits of HPFAS.
\endgroup
\end{abstract}


\section{Introduction}
The evolution toward beyond-5G/6G networks is characterized by increasingly heterogeneous propagation conditions and service requirements~\cite{10054381}, where a single antenna architecture is often insufficient to guarantee robust performance. In dense user-centric deployments, both strong LoS links and reliable operation under rich scattering are required, motivating flexible antenna designs capable of adapting to the radio environment.

Fluid antenna systems (FAS) have recently emerged as a promising paradigm that exploits multiple candidate ports within a compact region, enabling dynamic port selection to reduce spatial correlation and achieve diversity gains~\cite{9264694,11302793}. This flexibility further enables fluid antenna multiple access (FAMA), which simplifies interference management without precoding or interference cancellation~\cite{9650760,10066316}.

However, FAS is limited by its restricted spatial movement range and is mainly effective in rich scattering environments.  Dielectric waveguide-based pinching antennas have been proposed to create controllable radiation points along a guiding structure, enabling the formation of strong user-centric LoS links with low-cost and scalable hardware \cite{10945421}. By adjusting the positions or activation of pinching elements, the transceiver can effectively reshape the link geometry, which is particularly attractive for corridor-like deployments, hotspot coverage, and high-frequency systems. Recent works have demonstrated significant performance gains by optimizing pinching antenna locations or refining antenna positions in both single-user and multi-user downlink scenarios \cite{10896748,10981775,11016750}.

{\color{black}Motivated by these complementary characteristics, this letter proposes a hybrid pinching–fluid antenna system (HPFAS), where pinching antennas are utilized to enhance LoS links,  while fluid antennas exploit spatial diversity via port selection to mitigate small-scale fading.} {\color{red}The proposed HPFAS serves as an analytical framework for studying the
	joint effect of transmit-side link-geometry adaptation and receive-side
	fading diversity.} The main contributions of this letter are summarized as follows:
\begin{itemize} 
	\item We propose a HPFAS that jointly exploits pinching antenna (PA)-enabled LoS
	link enhancement and fluid-antenna (FA)-assisted spatial diversity.
	\item We develop a tractable channel model that captures the coexistence of pinching-induced LoS components and spatially varying small-scale fading at the fluid antenna, enabling analysis under Rician fading. 
	\item We derive analytical outage probability expressions and their
	approximations,
	\begingroup
	\color{red}
	as well as high-SNR asymptotic results,
	\endgroup
	and validate them via Monte Carlo simulations. The proposed approximation
	significantly reduces computational complexity while maintaining
	satisfactory accuracy.
	\begingroup
	\color{red}
	Numerical comparisons with PA-only and FA-only schemes are further
	provided to clarify the roles of the two mechanisms and offer
	scheme-selection insights.
	\endgroup
\end{itemize}
\section{System Model and Performance Analysis for Single-User HPFAS}
As illustrated in \Cref{fig:model}, we consider a downlink single-user HPFAS where a base station (BS), equipped with a single PA on a waveguide, serves a user equipment (UE) equipped with a linear FA of length $W$ with $N$ ports\footnote{\color{black}The single-user model will be extended to a general multi-user scenario in the next section.}. The transmitter can dynamically adjust the PA position along the waveguide, while the receiver adaptively selects one FA port to improve link quality.
A three-dimensional Cartesian coordinate system is adopted to describe the network geometry. The waveguides are aligned parallel to the $x$-axis at height $h$. The UE is located at $\boldsymbol{\psi}=(\phi_x,\phi_y,0)$, where $(\phi_x,\phi_y)$ is uniformly distributed over a $D_1 \times D_2$ rectangular area on the ground, while the PA is located at $\boldsymbol{\psi}_{p}^{\text{Pin}}
=
(\phi_{x,p},0,h)$.
\subsection{Wireless channel model}
{\color{black}We consider a wireless propagation environment where the downlink channel consists of a deterministic LoS component and multiple NLoS components induced by randomly distributed scatterers, as commonly observed in urban, indoor, and mmWave scenarios. Accordingly, the overall channel is modeled as the superposition of LoS and NLoS components, capturing both the dominant direct path and multipath effects.
\begin{figure}[t]
	\centering
	\includegraphics[width=2.5in]{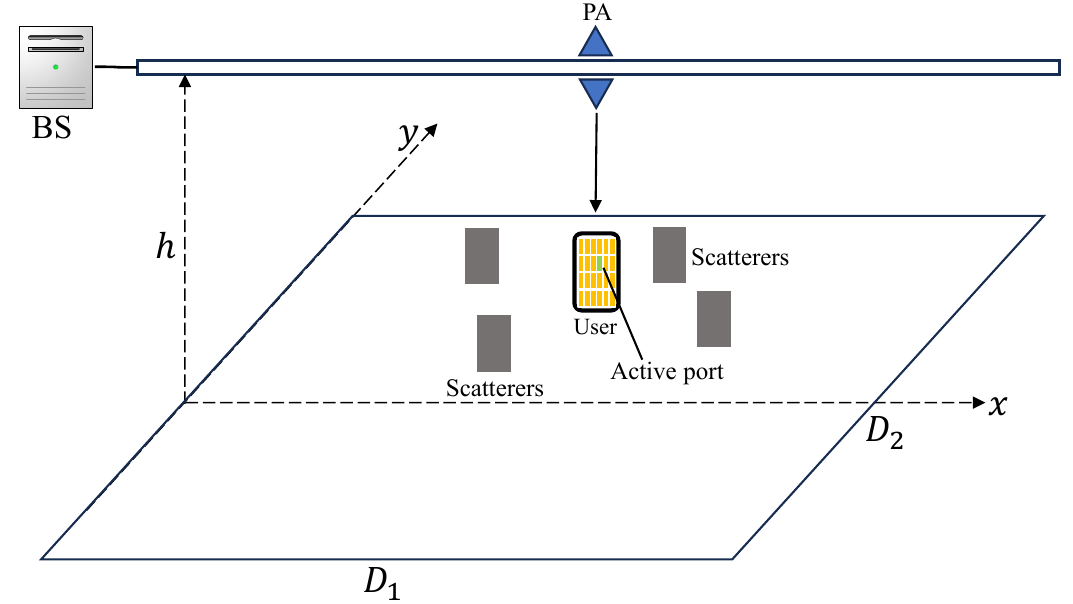}  
	\caption{System model of the proposed HPFAS {\color{black}in the single-user scenario.}}
	\label{fig:model}
\end{figure}
	Following the geometric spherical wavefront model~\cite{10945421}, the LoS channel between the PA and FA is given by\footnote{\color{black}The considered channel includes both waveguide transmission from the feeding point to the PA and subsequent free-space propagation to the UE.}
	\begin{equation}
		h^{\text{LoS}} 
		= \sqrt{\frac{\eta_0}{d^{\varepsilon}}}
		e^{-j\frac{2\pi d}{\lambda}}
		e^{-j\frac{2\pi d_0}{\lambda_g}},
	\end{equation}
	where $\eta_0 = \frac{\lambda^2}{(4\pi)^2}$, $\varepsilon$ is the path loss exponent, $d$ is the PA--FA distance, and $d_0$ is the distance between the PA and the waveguide feeding point. $\lambda$ and $\lambda_g=\lambda/n_{\text{eff}}$ denote the free-space and guided wavelengths, respectively.
	The normalized LoS phase term is defined as
	\begin{equation}
		\bar{h}^{\text{LoS}}
		= e^{-j\frac{2\pi d}{\lambda}}
		e^{-j\frac{2\pi d_0}{\lambda_g}},
	\end{equation}
	 Thus, $h^{\text{LoS}}=\sqrt{\beta(d)}\,\bar{h}^{\text{LoS}}$, where $d$  denotes the distance between the
	PA and the FA, and $\beta(d)=\eta_0 d^{-\varepsilon}$.
	
	Under the assumption of infinitely many scatterers, the NLoS component follows Rayleigh fading. Based on~\cite{10623405}, the NLoS channel between the PA and the $n$-th FA port is modeled as
	\begin{multline}
		h_{n,b(n)}^{\text{NLoS}} 
		= \sqrt{1-\mu^2}\, x_{n,b(n)} + \mu\, x_{b(n)} \\
		+ j\!\left(\sqrt{1-\mu^2}\, y_{n,b(n)} + \mu\, y_{b(n)}\right),
	\end{multline}
	where $\mu$ is close to $1$, and all Gaussian variables are i.i.d. with zero mean and unit variance. The index $b(n)$ denotes the correlation block structure with $\sum_{b=1}^{B} L_b = N$. Finally, the overall channel is modeled as a Rician channel, which is given by
	\begin{equation}
		h_{n,b(n)} 
		= \sqrt{\beta(d)}
		\left(
		\sqrt{\frac{\kappa}{\kappa+1}}\, \bar{h}^{\text{LoS}}
		+ \sqrt{\frac{1}{2(\kappa+1)}}\, h_{n,b(n)}^{\text{NLoS}}
		\right),
	\end{equation}
	where $\kappa$ denotes the Rician factor.}
\subsection{Performance Analysis of HPFAS}
The received signal $R$ at the $n$-th port of the UE is given by \eqref{signal},
	\begin{equation}\label{signal}
	\begin{aligned}
		R_n &= \sqrt{P_t}	h_{n,b(n)}s+z\\
	&	=\sqrt{P_t\beta(d)}
		\left(
		\sqrt{\frac{\kappa}{\kappa+1}}\, \bar{h}^{\text{LoS}}
		+ \sqrt{\frac{1}{2\left(\kappa+1\right)}}\, h_{n,b(n)}^{\text{NLoS}}
		\right)s+z
	\end{aligned}
\end{equation}
where $s \in \mathbb{C} $ is the downlink signal intended for UE having the unit power, \textit{i.e.},  $\mathbb{E}[\vert s\vert^2]=1$, and $z$ is the additive noise modeled as a circularly symmetric complex Gaussian random variable with zero mean and variance $\sigma^2$.  Under the alignment assumption $\phi_x,p = \phi_x$\footnote{\color{black}This alignment represents the optimal geometric configuration in the single-user case, as it minimizes the propagation distance and maximizes the effective channel gain, thereby achieving the highest SNR.}, the Euclidian distance reduces to $d=\sqrt{\phi_y^2+h^2}$. Then, the received signal-to-noise ratio (SNR) can be expressed as
\begin{equation}\label{gamma}
		\begin{aligned}
	\gamma_{n,b(n)}=\frac{P_t\eta_0\left\vert\sqrt{\frac{\kappa}{\kappa+1}}e^{-j\frac{2\pi d}{\lambda} -j\frac{2\pi d_0}{\lambda_g}  }+\sqrt{\frac{1}{2\left(\kappa+1\right)}}h_{n,b(n)}^{\text{NLoS}}\right\vert^2}{\sigma^2\left(\phi_y^2+h^2\right)^{\frac{\varepsilon}{2}}},
\end{aligned}
\end{equation}

 Accordingly, the outage probability is
defined as the probability that the maximum received SNR among all candidate
ports falls below a predefined threshold $\gamma_{\text{th}}$, yielding
\begin{equation}
	P_{\text{out}} = \text{Pr}\left(\max_n \gamma_{n,b(n)}<\gamma_{\text{th}}\right).
\end{equation}
After that, the outage probability of the proposed HPFAS is obtained in the following theorem.
{\color{black}\begin{theorem}
	The outage probability of the proposed HPFAS scheme is given by
	\begin{multline}\label{OP-FAPA}
		P_{\text{out}}	=\int_{0}^{\frac{D_2}{2}} \frac{2}{D_2}\prod_{b=1}^{B} \int_{0}^{\infty} \frac{1}{2} \exp\left(-\frac{r_b + \frac{2\kappa}{\mu^2}}{2}\right) I_0\left(\sqrt{\frac{2\kappa r_b}{\mu^2}}\right) \\
		\times\left[1 - Q_1\left(\sqrt{\frac{\mu^2 r_b}{1-\mu^2}}, \sqrt{C}\right)\right]^{L_b} dr_bd\phi_y,
	\end{multline}
	where $Q_1(\cdot,\cdot)$ denotes the first-order Marcum 
	$Q$-function and $C = \frac{2\left(\kappa+1\right)\gamma_{\text{th}}\sigma^2\left(\phi_y^2+h^2\right)^{\frac{\varepsilon}{2}}}{P_t\eta_0\left(1-\mu^2\right)}$. 
\end{theorem} 
\begin{IEEEproof}
We first define the random variable
	\begin{equation}\label{Xn}
		\Theta_n = \left(x_{n,b(n)} + \frac{\mu \tilde{x}_{b(n)}}{\sqrt{1-\mu^2}}\right)^2+\left(y_{n,b(n)} + \frac{\mu \tilde{y}_{b(n)}}{\sqrt{1-\mu^2}}\right)^2,
	\end{equation}
	where
	$\tilde{x}_{b(n)} \sim\mathcal{N}\left(\frac{\sqrt{2\kappa}}{\mu}\cos\left(\frac{2\pi d}{\lambda}+\frac{2\pi d_0}{\lambda_g}\right),1\right)$, $\tilde{y}_{b(n)}\sim\mathcal{N}\left(\frac{\sqrt{2\kappa}}{\mu}\sin\left(\frac{2\pi d}{\lambda}+\frac{2\pi d_0}{\lambda_g}\right),1\right)$.
	
Then, the outage probability can be rewritten as
\begin{equation}
		P_{\text{out}}= \text{Pr}\left(\max_n \Theta_n < \frac{2\left(\kappa+1\right)\gamma_{\text{th}}\sigma^2\left(\phi_y^2+h^2\right)^{\frac{\varepsilon}{2}}}{P_t\eta_0\left(1-\mu^2\right)}\right).
\end{equation}

 Let $r_b=\tilde{x}_b^2+\tilde{y}_b^2$, $\Theta_n$.  Conditioned on $r_b$, $\Theta_n$ follows a non-central chi-square distribution with two degrees of freedom. Its conditional
  probability density function  (PDF) is given by
\begin{equation}
	f_{\Theta_n|r_b}(\theta) = \frac{1}{2} \exp\left(-\frac{\theta + \frac{\mu^2}{1 - \mu^2}r_b}{2}\right) I_0\left(\sqrt{\frac{\mu^2 r_b \theta}{1 - \mu^2}}\right),
\end{equation}
where $I_0(\cdot)$ denotes the modified Bessel function of the first kind with order zero. Note that each $r_b$ for $b=1,\ldots,B$ is a non-central chi-square distribution with two
degrees of freedom, and its PDF is
 \begin{equation}
	f_{r_b}(r_b) = \frac{1}{2} \exp\left(-\frac{r_b + \frac{2\kappa}{\mu^2}}{2}\right) I_0\left(\sqrt{\frac{2\kappa}{\mu^2}r_b}\right), b=1,\ldots,B.
\end{equation}

Thus, the unconditional joint PDF of $\Theta_n$ can be derived as
\begin{multline}
	f_{\Theta_n}(\theta_1,\ldots,\theta_N)
	= \prod_{b=1}^{B} \int_{0}^{\infty}
	f_{r_b}(r_b)
	\prod_{n\in\mathcal{N}_b}
	f_{\Theta_n|r_b}(\theta_n)\, dr_b.
\end{multline}

The joint cumulative distribution function (CDF) 
 follows directly, and by averaging over the user location parameter $\phi_y$, the outage probability in \eqref{OP-FAPA} is obtained. This completes the proof.
\end{IEEEproof}}

 {\color{black}Although \eqref{OP-FAPA} provides an exact expression for the proposed HPFAS system, it involves a $(B\!+\!1)$-fold integral, which makes further analysis intractable. To gain more analytical insights, we adopt the step function approximation (SFA) method proposed in \cite[Lemma 1]{11184149} and \cite[Lemma 2]{11184149} to approximate the expression. Then, by applying the SFA, the outage probability in \eqref{OP-FAPA} can be approximated in the following lemma.

\begin{lemma}\label{HPFAS-approximate}
	The outage probability of the proposed HPFAS is approximated as
	\begin{equation}\label{appro}
		\begin{aligned}
			P_{\text{out}}
			=\int_{0}^{\frac{D_2}{2}} \frac{2}{D_2}\prod_{b=1}^{B}\left[ 1-Q_1\left(\sqrt{\frac{2\kappa}{\mu^2}},	\delta(\phi_y, L_b)\right)\right] d\phi_y,
		\end{aligned}
	\end{equation}
where $
	\delta(\phi_y, L_b) = \sqrt{\frac{1-\mu^2}{\mu^2}}\left[\sqrt{C} + \frac{\frac{L_b - 1}{\sqrt{2\pi}}\sqrt{C} +  \frac{1}{2}}{\frac{(L_b - 1)}{2\sqrt{2\pi}} + \frac{1}{2\sqrt{C}} - \sqrt{C}}\right]
$.
\end{lemma}
\begin{IEEEproof}
The result can be derived by following similar steps as in \cite[Appendix B]{11184149}. The detailed derivation is omitted for brevity.
\end{IEEEproof}

\begingroup
\color{red}

To further simplify the remaining finite-interval spatial integral, we
apply the Gauss--Legendre quadrature. By introducing the variable
transformation
\begin{equation}
	\phi_y=\frac{D_2}{4}(x+1),
\end{equation}
the integration interval $\phi_y\in[0,D_2/2]$ is mapped onto
$x\in[-1,1]$. Accordingly, \eqref{appro} can be rewritten as
\begin{equation}
	\begin{aligned}
		P_{\mathrm{out}}^{\mathrm{H,SFA}}
		=
		\frac{1}{2}
		\int_{-1}^{1}
		\prod_{b=1}^{B}
		\Bigg[
		1-Q_1\Bigg(
		\sqrt{\frac{2\kappa}{\mu^2}},
		\delta\left(
		\frac{D_2}{4}(x+1),L_b
		\right)
		\Bigg)
		\Bigg]
		dx.
	\end{aligned}
\end{equation}

Using the \textcolor{red}{$J$}-point Gauss--Legendre quadrature, the above integral is
approximated as
\begin{equation}\label{appro_GL}
	\begin{aligned}
		P_{\mathrm{out}}^{\mathrm{H,SFA}}
		=
		\frac{1}{2}
		\sum_{m=1}^{J} w_m
		\prod_{b=1}^{B}
		\Bigg[
		1-Q_1\Bigg(
		\sqrt{\frac{2\kappa}{\mu^2}},
		\delta\left(
		\frac{D_2}{4}(x_m+1),L_b
		\right)
		\Bigg)
		\Bigg].
	\end{aligned}
\end{equation}
where $x_m$ denotes the $m$th zero of the $J$th-order Legendre
polynomial, and $w_m$ is the corresponding quadrature weight.

\endgroup

\begingroup
\color{red}

\begin{remark}
	By applying the SFA, the $B$ auxiliary integrations associated with the
	FA correlation blocks are eliminated, reducing the original $(B+1)$-fold
	integral in \eqref{OP-FAPA} to the one-dimensional spatial integral in
	\eqref{appro}. The remaining integral is further converted into the finite
	weighted sum in \eqref{appro_GL} using Gauss--Legendre quadrature.
	Therefore, \eqref{appro_GL} is referred to as a semi-closed-form
	approximation.
\end{remark}

\endgroup
} 
%

\begin{corollary}\label{pro1}
	When the UE is equipped with a conventional antenna instead of the FA, the outage probability is given by
		\begin{equation}\label{OP-PA-only}
		P_{\text{out}}=1-\frac{2}{D_2}\int_{0}^{\frac{D_2}{2}} Q_1\left(\sqrt{2\kappa},\sqrt{C}\right)d\phi_y,
		\end{equation}
		where $C$ was previously defined.
	\end{corollary}
	\begin{IEEEproof}
	When the UE employs a conventional antenna, the random variable in \eqref{Xn} reduces to $\tilde{\Theta} = \tilde{x}^2+\tilde{y}^2$. This is because a conventional antenna can be viewed as a single-port FA,
	\textit{i.e.}, $N=1$, and thus $\tilde{x}_{b(n)}$ and $\tilde{y}_{b(n)}$ reduce to
	$\tilde{x}$ and $\tilde{y}$, respectively.
  Then, $\tilde{\Theta}$ follows a non-central chi-square distribution with two degrees of freedom.
Based on \cite[Eq.~(2.45)]{simon2002probability}, the CDF of $\tilde{\Theta}$ can be readily obtained. Finally, by averaging the conditional outage probability over the user location, the outage probability of the PA-only scheme is obtained in \eqref{OP-PA-only}, which completes the proof.
	\end{IEEEproof}
	\begingroup
	\color{red}
	
	Using the same $J$-point Gauss--Legendre quadrature,
	\eqref{OP-PA-only} can be approximated as
	\begin{equation}\label{OP-PA-only-GL}
		P_{\mathrm{out}}
		=
		1-\frac{1}{2}
		\sum_{m=1}^{J}
		w_m
		Q_1\left(
		\sqrt{2\kappa},
		\sqrt{C_m}
		\right),
	\end{equation}
	where
	\begin{equation}
		C_m=
		\frac{
			2(\kappa+1)\gamma_{\mathrm{th}}\sigma^2
			\left[
			\left(\frac{D_2}{4}(x_m+1)\right)^2+h^2
			\right]^{\varepsilon/2}
		}{
			P_t\eta_0(1-\mu^2)
		}.
	\end{equation}
	
	\endgroup
	\begin{corollary}
By fixing the antenna position at the geometric center of the planar region, when a conventional antenna is employed at the BS instead of a PA, the outage probability is given by
\begin{multline}
	P_{\text{out}}
	=\int_{-\frac{D_2}{2}}^{\frac{D_2}{2}}
	\int_{0}^{D_1}
	\frac{1}{D_1D_2}
	\prod_{b=1}^{B} \int_{0}^{\infty} \frac{1}{2} e^{-\frac{r_b + \frac{2\kappa}{\mu^2}}{2}} I_0\left(\sqrt{\frac{2\kappa r_b}{\mu^2}}\right) \\
	\times\left[1 - Q_1\left(\sqrt{\frac{\mu^2 r_b}{1-\mu^2}}, \sqrt{\tilde{C}}\right)\right]^{L_b} dr_bd\phi_xd\phi_y .
\end{multline}
where $\tilde{C}= \frac{2\left(\kappa+1\right)\gamma_{\text{th}}\sigma^2\left(\left(\phi_x-\frac{D_1}{2}\right)^2+\phi_y^2+h^2\right)^{\frac{\varepsilon}{2}}}{P_t\eta_0\left(1-\mu^2\right)}$.
\end{corollary}
\begin{IEEEproof}
When a conventional antenna is employed at the BS, the link distance becomes
$\sqrt{\left(\phi_x-\frac{D_1}{2}\right)^2+\phi_y^2+h^2}$. Consequently, the corresponding threshold term is modified to $\tilde{C}$, which concludes the proof.
\end{IEEEproof}
	{\color{black}
	\begin{lemma}
When a conventional antenna is employed at the BS instead of a PA, the outage probability  is approximated by
	\begin{equation}\label{OP-FA-only}
		\begin{aligned}
			P_{\text{out}}
			=	\int_{-\frac{D_2}{2}}^{\frac{D_2}{2}}
			&\int_{0}^{D_1}
			\frac{1}{D_1D_2}\times\\
			&\prod_{b=1}^{B}\left[ 1-Q_1\left(\sqrt{\frac{2\kappa}{\mu^2}},	\tilde{\delta}(\phi_x, \phi_y, L_b)\right)\right] d\phi_xd\phi_y, \\
		\end{aligned}
	\end{equation}
	where 
$
\tilde{\delta}(\phi_x,\phi_y, L_b) = \sqrt{\frac{1-\mu^2}{\mu^2}}\left[\sqrt{\tilde{C}} + \frac{\frac{L_b - 1}{\sqrt{2\pi}}\sqrt{\tilde{C}} +  \frac{1}{2}}{\frac{(L_b - 1)}{2\sqrt{2\pi}} + \frac{1}{2\sqrt{\tilde{C}}} - \sqrt{\tilde{C}}}\right]
$.
	\end{lemma}
	\begin{IEEEproof}
		The results can be easily obtained by following the same steps as in the proof of Lemma \ref{HPFAS-approximate}, and thus is omitted here.
	\end{IEEEproof}
	\begingroup
	\color{red}
	
	Applying the $J$-point Gauss--Legendre quadrature to the two
	finite-interval spatial integrations in \eqref{OP-FA-only} yields
	\begin{equation}\label{OP-FA-only-GL}
		P_{\mathrm{out}}
		=
		\frac{1}{4}
		\sum_{p=1}^{J}
		\sum_{m=1}^{J}
		w_pw_m
		\prod_{b=1}^{B}
		\left[
		1-Q_1\left(
		\sqrt{\frac{2\kappa}{\mu^2}},
		\widetilde{\delta}_{p,m}(L_b)
		\right)
		\right],
	\end{equation}
	where
	\begin{equation}
		\phi_{x,p}=\frac{D_1}{2}(x_p+1),
		\qquad
		\phi_{y,m}=\frac{D_2}{2}x_m,
	\end{equation}
	and
	\begin{equation}
		\widetilde{\delta}_{p,m}(L_b)
		=
		\widetilde{\delta}
		\left(\phi_{x,p},\phi_{y,m},L_b\right).
	\end{equation}
	Here, $x_p$, $x_m$, $w_p$, and $w_m$ denote the
	Gauss--Legendre nodes and corresponding weights.
	
	\endgroup
}
\begingroup
\color{red}
\begingroup
\color{red}
\subsection{High-SNR Asymptotic Analysis}
To characterize the high-SNR behavior, we analyze the
diversity orders and asymptotic outage coefficients of the
HPFAS, PA-only, and FA-only schemes as $P_t\rightarrow\infty$.
The following results are derived from the exact outage
probability expressions rather than the SFA-based
approximations.

\begin{lemma}
	For fixed and finite $\kappa$ and $\mu^2<1$, the high-SNR
	outage probability of scheme
	$i\in\{\mathrm{H},\mathrm{PA},\mathrm{FA}\}$ satisfies
	\begin{equation}
		P_{\mathrm{out}}^{i}
		\sim
		A_iP_t^{-d_i},
		\qquad P_t\rightarrow\infty,
	\end{equation}
	where $A_i$ and $d_i$ denote the asymptotic outage coefficient
	and diversity order, respectively. For HPFAS,
	\begin{equation}
		\begin{aligned}
			A_{\mathrm{H}}
			={}&
			\left[
			\frac{(\kappa+1)\gamma_{\mathrm{th}}\sigma^2}
			{\eta_0}
			\right]^N
			\frac{2}{D_2}
			\left(
			\prod_{b=1}^{B}\Lambda_b
			\right)
			\\
			&\times
			\int_{0}^{D_2/2}
			\left(\phi_y^2+h^2\right)^{N\varepsilon/2}
			d\phi_y,
		\end{aligned}
	\end{equation}
	where
	\begin{equation}
		\Lambda_b
		=
		\frac{
			(1-\mu^2)^{1-L_b}
		}{
			1+(L_b-1)\mu^2
		}
		\exp\left[
		-\frac{\kappa L_b}
		{1+(L_b-1)\mu^2}
		\right].
	\end{equation}
	
	For the PA-only scheme,
	\begin{equation}
		\begin{aligned}
			A_{\mathrm{PA}}
			={}&
			\frac{(\kappa+1)\gamma_{\mathrm{th}}\sigma^2}
			{\eta_0}
			e^{-\kappa}
			\frac{2}{D_2}
			\times
			\int_{0}^{D_2/2}
			\left(\phi_y^2+h^2\right)^{\varepsilon/2}
			d\phi_y.
		\end{aligned}
	\end{equation}
	
	For the FA-only scheme,
	\begin{equation}
		\begin{aligned}
			A_{\mathrm{FA}}
			={}&
			\left[
			\frac{(\kappa+1)\gamma_{\mathrm{th}}\sigma^2}
			{\eta_0}
			\right]^N
			\frac{1}{D_1D_2}
			\left(
			\prod_{b=1}^{B}\Lambda_b
			\right)
			\\
			&\times
			\int_{-D_2/2}^{D_2/2}
			\int_{0}^{D_1}
			\left[
			\left(\phi_x-\frac{D_1}{2}\right)^2
			+\phi_y^2+h^2
			\right]^{N\varepsilon/2}
			d\phi_x\,d\phi_y.
		\end{aligned}
	\end{equation}
	
	The corresponding diversity orders are
	\begin{equation}
		d_{\mathrm{H}}=N,\qquad
		d_{\mathrm{FA}}=N,\qquad
		d_{\mathrm{PA}}=1.
	\end{equation}
\end{lemma}

\begin{IEEEproof}
	As $P_t\rightarrow\infty$, the threshold terms in the exact
	outage probability expressions approach zero. For fixed $a$,
	the first-order Marcum $Q$-function satisfies
	\begin{equation}
		1-Q_1(a,\sqrt{x})
		\sim
		\frac{x}{2}
		\exp\left(-\frac{a^2}{2}\right),
		\qquad x\rightarrow0.
	\end{equation}
	
	Let $I_b^{\mathrm{H}}(\phi_y)$ denote the inner integral
	associated with the $b$th FA correlation block in the exact
	HPFAS outage probability. Applying the above expansion gives
	\begin{equation}
		\begin{aligned}
			I_b^{\mathrm{H}}(\phi_y)
			\sim{}&
			\left(\frac{C}{2}\right)^{L_b}
			\frac{1}{2}
			\exp\left(-\frac{\kappa}{\mu^2}\right)
			\\
			&\times
			\int_{0}^{\infty}
			\exp\left[
			-\frac{
				1+(L_b-1)\mu^2
			}{
				2(1-\mu^2)
			}r_b
			\right]
			I_0\left(
			\sqrt{\frac{2\kappa r_b}{\mu^2}}
			\right)
			dr_b.
		\end{aligned}
	\end{equation}
	
	Evaluating the above integral and substituting the definition
	of $C$ yield
	\begin{equation}
		I_b^{\mathrm{H}}(\phi_y)
		\sim
		\Lambda_b
		\left[
		\frac{
			(\kappa+1)\gamma_{\mathrm{th}}\sigma^2
			(\phi_y^2+h^2)^{\varepsilon/2}
		}{
			P_t\eta_0
		}
		\right]^{L_b}.
	\end{equation}
	
	Therefore,
	\begin{equation}
		\begin{aligned}
			\prod_{b=1}^{B}I_b^{\mathrm{H}}(\phi_y)
			\sim{}&
			\left(
			\prod_{b=1}^{B}\Lambda_b
			\right)
			\left[
			\frac{
				(\kappa+1)\gamma_{\mathrm{th}}\sigma^2
			}{
				P_t\eta_0
			}
			\right]^{\sum_{b=1}^{B}L_b}
			\\
			&\times
			\left(\phi_y^2+h^2\right)^{
				\frac{\varepsilon}{2}
				\sum_{b=1}^{B}L_b
			}.
		\end{aligned}
	\end{equation}
	
	Since
	\begin{equation}
		\sum_{b=1}^{B}L_b=N,
	\end{equation}
	averaging over the user location gives
	\begin{equation}
		P_{\mathrm{out}}^{\mathrm{H}}
		\sim
		A_{\mathrm{H}}P_t^{-N}.
	\end{equation}
	
	For the PA-only scheme, the complement of the first-order
	Marcum $Q$-function appears only once. Applying the same
	small-threshold expansion and averaging over the user location
	give
	\begin{equation}
		P_{\mathrm{out}}^{\mathrm{PA}}
		\sim
		A_{\mathrm{PA}}P_t^{-1}.
	\end{equation}
	
	For the FA-only scheme, replacing the HPFAS
	distance-dependent term by the FA-only distance-dependent term
	leads to the same block-level high-SNR exponent. Hence,
	\begin{equation}
		P_{\mathrm{out}}^{\mathrm{FA}}
		\sim
		A_{\mathrm{FA}}P_t^{-N}.
	\end{equation}
	
	Using
	\begin{equation}
		d_i
		=
		-\lim_{P_t\rightarrow\infty}
		\frac{\log P_{\mathrm{out}}^i}
		{\log P_t},
	\end{equation}
	the stated diversity orders follow.
\end{IEEEproof}

\begin{remark}
	The asymptotic results can be summarized as
	\begin{equation}
		d_{\mathrm{H}}
		=
		d_{\mathrm{FA}}
		=
		N
		>
		d_{\mathrm{PA}}
		=
		1,
		\qquad
		A_{\mathrm{H}}
		\leq
		A_{\mathrm{FA}}.
	\end{equation}
	The diversity-order result shows that the $N$ FA ports
	determine the diversity order of HPFAS, whereas the PA does
	not provide an additional statistically independent diversity
	branch. Moreover, for every user location,
	\begin{equation}
		\phi_y^2+h^2
		\leq
		\left(\phi_x-\frac{D_1}{2}\right)^2
		+\phi_y^2+h^2,
	\end{equation}
	which leads to
	$A_{\mathrm{H}}\leq A_{\mathrm{FA}}$.
	Therefore, although HPFAS and FA-only achieve the same
	diversity order, the PA reduces the asymptotic outage
	coefficient through link-geometry adjustment. A smaller
	asymptotic outage coefficient corresponds to a lower outage
	probability in the high-SNR regime.
\end{remark}

\endgroup
\endgroup

\section{System Model and Performance Analysis for Multi-User HPFAS}
\subsection{Wireless Signal Model}
Consider a base station equipped with $K$ dielectric waveguides, each supporting
$M$ pinching antennas, serving multiple single-antenna users via
hybrid-antenna multiple access.
Each waveguide transmits an independent data stream dedicated to one user.
Without loss of generality, we assume that the total transmit power allocated
to each waveguide is identical and is denoted by $P$. The received signal at the $k$-th user is given by
\begin{equation}
	r_{k} =
	\underbrace{\sqrt{P}\sum_{m=1}^{M} h_{n,b(n)}^{k,m} w_{k}^{m} s_{k}}_{\text{desired signal}}
	+ \underbrace{\sqrt{P}\sum_{i\neq k}\sum_{m=1}^{M} h_{n,b(n)}^{i,m} w_{i}^{m} s_{i}}_{\text{inter-user interference}}
	+ z_{k},
\end{equation}
where $h_{n,b(n)}^{k,m}$ is the corresponding channel coefficient between $m$-th PA and the $n$-th port  of the $k$-th user, $w_k^m$ denotes the beamforming coefficient associated with the $m$-th PA on the $k$-th waveguide, satisfying the unit-norm constraint $	\sum_{m=1}^{M} |w_k^m|^2 = 1$, and $z_k$ represents the additive white Gaussian noise at the $k$-th user. Accordingly, the signal-to-interference ratio (SIR)\footnote{We use the signal-to-interference ratio (SIR) instead of the signal-to-noise ratio (SNR) by assuming an interference-limited scenario, where the interference power dominates the noise power.} of the $k$-th user is given by
\begin{equation}\label{SIRmultiple}
	\text{SIR}_k =
	\frac{\left|\sum_{m=1}^{M} h_{n,b(n)}^{k,m} w_{k}^{m}\right|^{2}}
	{\left|\sum_{i\neq k}\sum_{m=1}^{M} h_{n,b(n)}^{i,m} w_{i}^{m}\right|^{2}}.
\end{equation}
Although \eqref{SIRmultiple} provides a general performance metric, an analytically tractable expression for the outage probability with arbitrary $K$ and $M$ is unavailable. Therefore, we focus on a representative special case to gain fundamental insights.
\subsection{Special Case: Two Waveguides, Two Users, and One Antenna per Waveguide ($K=2$, $M=1$)}
We consider a simplified scenario with two waveguides, two users and one PA for each waveguide. The received signals at users~1 and~2 are respectively given by
\begin{equation}
	r_{1,n} =
	\underbrace{\sqrt{P} h_{n,b(n)}^{1} s_{1}}_{\text{desired signal}}
	+ \underbrace{\sqrt{P} h_{n,b(n)}^{2} s_{2}}_{\text{interference}}
	+ z_{1},
\end{equation}
\begin{equation}
	r_{2,n} =
	\underbrace{\sqrt{P} h_{n,b(n)}^{2} s_{2}}_{\text{desired signal}}
	+ \underbrace{\sqrt{P} h_{n,b(n)}^{1} s_{1}}_{\text{interference}}
	+ z_{2}.
\end{equation}
where $s_1, s_2 \in \mathbb{C}$ denote the downlink signals intended for the UEs, each with unit power, \textit{i.e.}, $\mathbb{E}\!\left[|s_1|^2\right] = \mathbb{E}\!\left[|s_2|^2\right] = 1$, and $z_1$ and $z_2$ represent the additive noise terms. 
The planar area is then partitioned into two disjoint regions, namely Region~1 and Region~2, where user~1 and user~2 are uniformly distributed in Region~1 and Region~2, respectively. 
The locations of user~1 and user~2 are denoted by $\boldsymbol{\psi}_1=(\phi_{x_1},\phi_{y_1},0)$ and $\boldsymbol{\psi}_2=(\phi_{x_2},\phi_{y_2},0)$, where $\phi_{x_1}, \phi_{x_2}\in(0,D_1)$, $\phi_{y_1}\in(-\frac{D_2}{2},0)$, and $\phi_{y_2}\in(0,\frac{D_2}{2})$. 
The corresponding pinching antennas are assumed to be deployed in Region~1 and Region~2, with locations given by $\boldsymbol{\psi}_1^{\text{Pin}}=(\phi_{x_1,p},-\frac{D_2}{4},h)$ and $\boldsymbol{\psi}_2^{\text{Pin}}=(\phi_{x_2,p},\frac{D_2}{4},h)$, respectively. 
The resulting SIR at user~1 for the $n$-th port realization is expressed as
\begin{equation}
	\begin{aligned}
	&\text{SIR}_{1,n}\\
	&=\frac{
		\left|\boldsymbol{\psi}_1-\boldsymbol{\psi}_2^{\text{Pin}}\right|^{\varepsilon}
		\left|\sqrt{2\kappa} e^{-j\frac{2\pi d_1}{\lambda_g}-j\frac{2\pi}{\lambda}|\boldsymbol{\psi}_1-\boldsymbol{\psi}_1^{\text{Pin}}|}
		+h_{n,b(n)}^{(1,\text{NLoS})}\right|^{2}
	}{
		\left|\boldsymbol{\psi}_1-\boldsymbol{\psi}_1^{\text{Pin}}\right|^{\varepsilon}
		\left|\sqrt{2\kappa} e^{-j\frac{2\pi d_2}{\lambda_g}-j\frac{2\pi}{\lambda}|\boldsymbol{\psi}_1-\boldsymbol{\psi}_2^{\text{Pin}}|}
		+h_{n,b(n)}^{(2,\text{NLoS})}\right|^{2}
	}.			
	\end{aligned}
\end{equation}
{\color{black}where $d_1$ denotes the distance between the PA and the waveguide feeding point of the first waveguide, and $d_2$ denotes that of the second waveguide.}
Without loss of generality, we focus on analyzing the outage probability of user~1, and the obtained results are equally applicable to user~2. Accordingly, the outage probability of user~1 is defined as
\begin{equation}
	P_{\text{out}}
	=\Pr\!\left(\max_{n}\text{SIR}_{1,n}<\gamma_{\text{th}}\right).
\end{equation}
We assume that $\phi_{x_1,p}=\phi_{x_1}$ and $\phi_{x_2,p}=\phi_{x_2}$\footnote{This assumption is introduced to decouple the spatial dependence among different users and to facilitate a tractable performance analysis.}. Under this assumption, the outage probability of the proposed HPFAS is given in the following theorem.
\begin{theorem}
	The outage probability of the proposed HPFAS scheme in the multi-user scenario is given by \eqref{OP-MU-FA-PA},	where $\widetilde{\gamma}_{\text{th}}
	=\frac{|\boldsymbol{\psi}_1-\boldsymbol{\psi}_1^{\text{Pin}}|^{\varepsilon}}
	{|\boldsymbol{\psi}_1-\boldsymbol{\psi}_2^{\text{Pin}}|^{\varepsilon}}\gamma_{\text{th}}$. 
	\begin{figure*}[t]
	\begin{equation}\label{OP-MU-FA-PA}
	\begin{aligned}
		&	P_{\text{out}}=
		\int_{\phi_{x_1}=0}^{D_1}
		\int_{\phi_{x_2}=0}^{D_1}
		\int_{\phi_{y_1}=-\frac{D_2}{2}}^{0}\frac{2}{D_1^2D_2}
		\prod_{b=1}^{B} \int_{\widetilde{r}_b=0}^{\infty}\int_{r_b=0}^{\infty} \frac{1}{4} \exp\left(-\frac{r_b + \frac{2\kappa}{\mu^2}}{2}\right) I_0\left(\sqrt{\frac{2\kappa r_b}{\mu^2}}\right)  \exp\left(-\frac{\widetilde{r}_b + \frac{2\kappa}{\mu^2}}{2}\right) I_0\left(\sqrt{\frac{2\kappa}{\mu^2}\widetilde{r}_b}\right)\\
		&\times\left[Q_1\!\left(
		\sqrt{\frac{\mu^2 \widetilde{\gamma}_{\text{th}} \widetilde{r}_b}{(1-\mu^2)(\widetilde{\gamma}_{\text{th}}+1)}},
		\sqrt{\frac{\mu^2 r_b}{(1-\mu^2)(\widetilde{\gamma}_{\text{th}}+1)}}
		\right) 
		- \frac{e^{	-\frac{\mu^2(\widetilde{\gamma}_{\text{th}} \widetilde{r}_b + r_b)}{2(1-\mu^2)(\widetilde{\gamma}_{\text{th}}+1)}}}{\widetilde{\gamma}_{\text{th}}+1}
		I_0\!\left(
		\frac{\mu^2 \sqrt{\widetilde{\gamma}_{\text{th}} r_b \widetilde{r}_b}}{(1-\mu^2)(\widetilde{\gamma}_{\text{th}}+1)}
		\right)\right]^{L_b}dr_bd\widetilde{r}_bd\phi_{x_1}d\phi_{x_2}d\phi_{y_1}\\
	\end{aligned}
\end{equation}
		\hrulefill
	\end{figure*}
\end{theorem}
\begin{IEEEproof}
	We first define the random variables $\vartheta_n$ and $\varrho_n$, which represent the effective desired signal power and interference power, respectively, as
	\begin{equation}
		\begin{aligned}
			\vartheta_n &=
			\left(
			x_{n,b(n)}^{(1)} + \frac{\mu \tilde{x}_{b(n)}^{(1)}}{\sqrt{1-\mu^2}}
			\right)^2
			+
			\left(
			y_{n,b(n)}^{(1)} + \frac{\mu \tilde{y}_{b(n)}^{(1)}}{\sqrt{1-\mu^2}}
			\right)^2, \\
			\varrho_n &=
			\left(
			x_{n,b(n)}^{(2)} + \frac{\mu \tilde{x}_{b(n)}^{(2)}}{\sqrt{1-\mu^2}}
			\right)^2
			+
			\left(
			y_{n,b(n)}^{(2)} + \frac{\mu \tilde{y}_{b(n)}^{(2)}}{\sqrt{1-\mu^2}}
			\right)^2 ,
		\end{aligned}
	\end{equation}
where, for $i\in\{1,2\}$,
$\left[\tilde{x}_{b(n)}^{(i)},\tilde{y}_{b(n)}^{(i)}\right]^{\mathsf T}$
follows a complex Gaussian distribution with mean
\begin{equation*}
\frac{\sqrt{2\kappa}}{\mu}
\begin{bmatrix}
	\cos\!\left(\frac{2\pi}{\lambda}\left|\boldsymbol{\psi}_1-\boldsymbol{\psi}_i^{\text{Pin}}\right|+\frac{2\pi d_1}{\lambda_g}\right)\\
	\sin\!\left(\frac{2\pi}{\lambda}\left|\boldsymbol{\psi}_1-\boldsymbol{\psi}_i^{\text{Pin}}\right|+\frac{2\pi d_2}{\lambda_g}\right)
\end{bmatrix}
\end{equation*} and unit variance. Then, the outage probability can be rewritten as
	\begin{equation}
		P_{\text{out}}=
		\Pr\!\left(
		\max_{n}
		\frac{\vartheta_n}{\varrho _n}
		<\widetilde{\gamma}_{\text{th}}
		\right).
	\end{equation}
	After that, the results can be obtained by following the similar steps as
	in~\cite[Appendix B]{10623405}. The detailed derivation is omitted for brevity.
\end{IEEEproof}

 {\color{black}It can be observed that \eqref{OP-MU-FA-PA} is highly complex, which hinders the extraction of meaningful insights. To address this issue, we adopt the SFA method to approximate the expression again, thereby obtaining a more concise and insightful result. After that, the outage probability can be approximated as stated in the following lemma.
\begin{lemma}
The outage probability of the proposed HPFAS scheme for the two-user case can be approximated as given in \eqref{OP-MU-FA-PA-approximate}. The threshold function is 
	\begin{figure*}[t]
	\begin{equation}\label{OP-MU-FA-PA-approximate}
	P_{\text{out}}=\int_{\phi_{x_1}=0}^{D_1}\int_{\phi_{x_2}=0}^{D_1}\int_{\phi_{y_1}=-\frac{D_2}{2}}^{0}\frac{2}{D_1^2D_2}\prod_{b=1}^{B} \int_{0}^{\infty}  \frac{1}{2} \exp\left(-\frac{r_b + \frac{2\kappa}{\mu^2}}{2}\right) I_0\left(\sqrt{\frac{2\kappa}{\mu^2}r_b}\right)Q_1\left(\sqrt{\frac{2\kappa}{\mu^2}}, \sqrt{\tilde{\delta}\left(r_b,L_b\right)}\right)dr_bd\phi_{x_1}d\phi_{x_2}d\phi_{y_1}
\end{equation}
	\hrulefill
\end{figure*}
$
\tilde{\delta}\left(r_b,L_b\right) = \frac{\left(1-\mu^2\right)}{\mu^2} \left[ \delta(r_b) + \frac{\frac{L_b - 1}{\sqrt{2\pi}} + \frac{1}{2\delta(r_b)}}{\frac{L_b - 1}{\sqrt{2\pi}} \frac{1}{2\delta(r_b)} + 2} \right]^2,
$
and
$
\delta(r_b) = \sqrt{ \frac{\mu^{2} r_b}{2(1 - \mu^{2}) \widetilde{\gamma}_{\text{th}} } }+ \sqrt{ \frac{\mu^{2} r_b}{2(1 - \mu^{2}) \widetilde{\gamma}_{\text{th}}} + \frac{1}{2 \widetilde{\gamma}_{\text{th}} }}.
$
\end{lemma}
\begin{IEEEproof}
\eqref{OP-MU-FA-PA} can be rewritten as shown in \eqref{OP-MU-FA-PA-approximate-proof}, where 
 $\left(a\right)$ follows from similar derivations as in \eqref{OP-FAPA}, $\left(b\right)$ is obtained by applying SFA,  $\left(c\right)$ follows from the definition of the generalized Marcum $Q$-function, and $\left(d\right)$ is derived by applying SFA again. Finally, \eqref{OP-MU-FA-PA-approximate} is obtained, and the proof ends.
\end{IEEEproof}
}
{\color{black}
\begin{corollary}
	When both users~1 and~2 are equipped with conventional antennas rather than the FA, the outage probability is given by
		\begin{equation}\label{OP-PA-only-final-new}
	\begin{aligned}
		P_{\text{out}}
		&=\int_{0}^{D_1}\int_{0}^{D_1}\int_{-\frac{D_2}{2}}^{0}\frac{2}{D_1^2D_2}
		\left[Q_{1}\left( \frac{\sqrt{2\kappa\widetilde{\gamma}_{\text{th}}} }{\sqrt{1+\widetilde{\gamma}_{\text{th}}}},  \sqrt{\frac{2\kappa}{1+\widetilde{\gamma}_{\text{th}}}} \right)\right.\\
		&\left.-\left( \frac{\exp\left(-\kappa\right)}{1+\widetilde{\gamma}_{\text{th}}} \right)I_0\left(\frac{2\kappa\sqrt{\widetilde{\gamma}_{\text{th}}}}{1+\widetilde{\gamma}_{\text{th}}}\right)\right]
		d\phi_{x_1}d\phi_{x_2}d\phi_{y_1}.
		\end{aligned}
	\end{equation}
\end{corollary}
\begin{IEEEproof}
	Similarly, the outage probability is given by
	\begin{multline}\label{OP-PA-only-final}
		P_{\text{out}}
		=1-\int_{0}^{D_1}\int_{0}^{D_1}\int_{-\frac{D_2}{2}}^{0}\frac{2}{D_1^2D_2}\int_{0}^{\infty}\exp\left(-\kappa\right)\exp\left(-r\right) \\
		\times Q_1\left(\sqrt{ 2\kappa}, \sqrt{2\widetilde{\gamma}_{\text{th}}r}\right) I_0\left(2\sqrt{ \kappa r}\right)drd\phi_{x_1}d\phi_{x_2}d\phi_{y_1}.
	\end{multline}
	To evaluate the \eqref{OP-PA-only-final}, we use an important connection established in \cite{6803901} as
	\begin{multline}
		\int_{0}^{\infty} \exp\left(-t\right) Q_{1}(\alpha, \beta\sqrt{t}) I_{0}\left(c\sqrt{t}\right)dt\\=\exp\left(\frac{c^2}{4}\right)-\text{In}\left(\beta, \alpha, c, 1, 0, 1\right).
	\end{multline}
	where  the result of $\text{In}\left(\beta, \alpha, c, 1, 0, 1\right)$ can be obtained by using the results in \cite[Eq. (24)]{6803901} as
	\begin{multline}		
	\text{In}\left(\beta, \alpha, c, 1, 0, 1\right)=\exp\left(\kappa\right) Q_{1}\left( \frac{\beta c}{\sqrt{2\left(2+\beta^2\right)}}, \alpha \sqrt{\frac{2}{2+\beta^2}} \right)\\
	-2 \left( \frac{1}{2+\beta^2} \right) \exp\left( \frac{\frac{c^2}{2} - \alpha^2}{2+\beta^2} \right)I_0\left(\frac{c\alpha\beta}{2+\beta^2}\right)\\
	=\exp\left(\kappa\right) Q_{1}\left( \frac{\sqrt{2\kappa\widetilde{\gamma}_{\text{th}}} }{\sqrt{1+\widetilde{\gamma}_{\text{th}}}},  \sqrt{\frac{2\kappa}{1+\widetilde{\gamma}_{\text{th}}}} \right)- \frac{I_0\left(\frac{2\kappa\sqrt{\widetilde{\gamma}_{\text{th}}}}{1+\widetilde{\gamma}_{\text{th}}}\right)}{1+\widetilde{\gamma}_{\text{th}}}.
	\end{multline}
	Finally, by substituting $\alpha = \sqrt{2\kappa}$, $\beta = \sqrt{2\widetilde{\gamma}_{\text{th}}}$, and $c = 2\sqrt{\kappa}$, \eqref{OP-PA-only-final-new} is obtained, which completes the proof.
\end{IEEEproof}
\begingroup
\color{red}

Applying the $J$-point Gauss--Legendre quadrature to the three
finite-interval spatial integrations in
\eqref{OP-PA-only-final-new} yields
\begin{equation}\label{OP-PA-only-final-GL}
	\begin{aligned}
		P_{\mathrm{out}}
		={}&
		\frac{1}{8}
		\sum_{i=1}^{J}
		\sum_{j=1}^{J}
		\sum_{\ell=1}^{J}
		\omega_i\omega_j\omega_\ell
		\Bigg[
		Q_1\left(
		\sqrt{
			\frac{2\kappa\widetilde{\gamma}_{i,j,\ell}}
			{1+\widetilde{\gamma}_{i,j,\ell}}
		},
		\sqrt{
			\frac{2\kappa}
			{1+\widetilde{\gamma}_{i,j,\ell}}
		}
		\right)
		\\
		&\qquad
		-
		\frac{\exp(-\kappa)}
		{1+\widetilde{\gamma}_{i,j,\ell}}
		I_0\left(
		\frac{
			2\kappa\sqrt{\widetilde{\gamma}_{i,j,\ell}}
		}{
			1+\widetilde{\gamma}_{i,j,\ell}
		}
		\right)
		\Bigg],
	\end{aligned}
\end{equation}
where $\widetilde{\gamma}_{i,j,\ell}$ denotes
$\widetilde{\gamma}_{\mathrm{th}}$ evaluated at
\begin{equation}
	\phi_{x_1}=\frac{D_1}{2}(\xi_i+1),\qquad
	\phi_{x_2}=\frac{D_1}{2}(\xi_j+1),\qquad
	\phi_{y_1}=\frac{D_2}{4}(\xi_\ell-1),
\end{equation}
and $\xi_i$ and $\omega_i$ denote the Gauss--Legendre nodes
and corresponding weights, respectively.

\endgroup
}
	\begin{figure*}[t]
	\begin{equation}\label{OP-MU-FA-PA-approximate-proof}
		\begin{aligned}
			&P_{\text{out}}
			\overset{\left(a\right)}{=}\int_{\phi_{x_1}=0}^{D_1}\int_{\phi_{x_2}=0}^{D_1}\int_{\phi_{y_1}=-\frac{D_2}{2}}^{0}\frac{2}{D_1^2D_2}		\prod_{b=1}^{B} \int_{\widetilde{r}_b=0}^{\infty}\int_{r_b=0}^{\infty} \frac{1}{2} \exp\left(-\frac{r_b + \frac{2\kappa}{\mu^2}}{2}\right) I_0\left(\sqrt{\frac{2\kappa r_b}{\mu^2}}\right) \frac{1}{2} \exp\left(-\frac{\widetilde{r}_b + \frac{2\kappa}{\mu^2}}{2}\right) I_0\left(\sqrt{\frac{2\kappa\widetilde{r}_b}{\mu^2}}\right)\\
			& \left[\int_{y_n=0}^{\infty}\left[1 - Q_1\left(\sqrt{\frac{\mu^2 r_b}{1-\mu^2}}, \sqrt{\widetilde{\gamma}_{\text{th}}y_n}\right)\right] \times\frac{1}{2} \exp\left(-\frac{y_n + \frac{\mu^2}{1-\mu^2}\widetilde{r}_b}{2}\right)  I_0\left(\sqrt{\frac{\mu^2 \widetilde{r}_b y_n}{1-\mu^2}}\right)dy_n\right]^{L_b}dr_bd\widetilde{r}_bd\phi_{x_1}d\phi_{x_2}d\phi_{y_1}\\
			&\overset{\left(b\right)}{=}\int_{\phi_{x_1}=0}^{D_1}\int_{\phi_{x_2}=0}^{D_1}\int_{\phi_{y_1}=-\frac{D_2}{2}}^{0}	\frac{2}{D_1^2D_2}\prod_{b=1}^{B} \int_{\widetilde{r}_b=0}^{\infty}\int_{r_b=0}^{\infty} \frac{1}{2} \exp\left(-\frac{r_b + \frac{2\kappa}{\mu^2}}{2}\right) I_0\left(\sqrt{\frac{2\kappa r_b}{\mu^2}}\right) \frac{1}{2} \exp\left(-\frac{\widetilde{r}_b + \frac{2\kappa}{\mu^2}}{2}\right) I_0\left(\sqrt{\frac{2\kappa\widetilde{r}_b}{\mu^2}}\right)\\
			&\qquad\qquad\qquad\qquad\qquad\qquad\times \left[\int_{y_n=\frac{\mu^2r_b}{\left(1-\mu^2\right)\widetilde{\gamma}_{\text{th}}} }^{\infty} \frac{1}{2}\exp\left(-\frac{y_n + \frac{\mu^2}{1-\mu^2}\widetilde{r}_b}{2}\right)  I_0\left(\sqrt{\frac{\mu^2 \widetilde{r}_b y_n}{1-\mu^2}}\right)dy_n\right]^{L_b}dr_bd\widetilde{r}_bd\phi_{x_1}d\phi_{x_2}d\phi_{y_1}\\
			&\overset{\left(c\right)}{=}\int_{\phi_{x_1}=0}^{D_1}\int_{\phi_{x_2}=0}^{D_1}\int_{\phi_{y_1}=-\frac{D_2}{2}}^{0}	\frac{2}{D_1^2D_2}\prod_{b=1}^{B} \int_{\widetilde{r}_b=0}^{\infty}\int_{r_b=0}^{\infty}\frac{1}{2} \exp\left(-\frac{r_b + \frac{2\kappa}{\mu^2}}{2}\right) I_0\left(\sqrt{\frac{2\kappa r_b}{\mu^2}}\right) \\
			&\qquad\qquad\qquad\qquad\qquad\times\frac{1}{2} \exp\left(-\frac{\widetilde{r}_b + \frac{2\kappa}{\mu^2}}{2}\right) I_0\left(\sqrt{\frac{2\kappa}{\mu^2}\widetilde{r}_b}\right)\left[Q_1\left(\sqrt{\frac{\mu^2}{1-\mu^2}\widetilde{r}_b}, \sqrt{\frac{\mu^2r_b}{\left(1-\mu^2\right)\widetilde{\gamma}_{\text{th}}} }\right) \right]^{L_b}dr_bd\widetilde{r}_bd\phi_{x_1}d\phi_{x_2}d\phi_{y_1}\\
			&\overset{\left(d\right)}{=} \int_{\phi_{x_1}=0}^{D_1}\int_{\phi_{x_2}=0}^{D_1}\int_{\phi_{y_1}=-\frac{D_2}{2}}^{0}\frac{2}{D_1^2D_2}\prod_{b=1}^{B} \int_{r_b=0}^{\infty} \int_{\widetilde{r}_b=	\tilde{\delta}\left(r_b,L_b\right)}^{\infty} \frac{1}{2} \exp\left(-\frac{r_b + \frac{2\kappa}{\mu^2}}{2}\right) I_0\left(\sqrt{\frac{2\kappa r_b}{\mu^2}}\right)\\
			&\qquad\qquad\qquad\qquad\qquad\qquad\qquad\qquad\qquad\qquad\qquad\qquad \times\frac{1}{2} \exp\left(-\frac{\widetilde{r}_b + \frac{2\kappa}{\mu^2}}{2}\right) I_0\left(\sqrt{\frac{2\kappa}{\mu^2}\widetilde{r}_b}\right)dr_bd\widetilde{r}_bd\phi_{x_1}d\phi_{x_2}d\phi_{y_1}\\
		\end{aligned}
	\end{equation}
	\hrulefill
\end{figure*}
\begin{corollary}
When conventional antennas are employed at the BS instead of pinching antennas, it is assumed that the antenna serving user~1 is located at the center of Region~1, \textit{i.e.}, {\color{black}$\left(\frac{D_1}{2},-\frac{D_2}{4},h\right)$}, while the antenna serving user~2 is located at the center of Region~2, \textit{i.e.}, {\color{black}$\left(\frac{D_1}{2},\frac{D_2}{4},h\right)$}.
the outage probability is given by \eqref{OP-MU-FA-Only},
	\begin{figure*}[t]
	\begin{equation}\label{OP-MU-FA-Only}
	\begin{aligned}
		&	P_{\text{out}}=
		\int_{\phi_{x_1}=0}^{D_1}\int_{\phi_{y_1}=-\frac{D_2}{2}}^{0}	\frac{2}{D_1D_2}	\prod_{b=1}^{B} \int_{\widetilde{r}_b=0}^{\infty}\int_{r_b=0}^{\infty} \frac{1}{4} \exp\left(-\frac{r_b + \frac{2\kappa}{\mu^2}}{2}\right) I_0\left(\sqrt{\frac{2\kappa r_b}{\mu^2}}\right)  \exp\left(-\frac{\widetilde{r}_b + \frac{2\kappa}{\mu^2}}{2}\right) I_0\left(\sqrt{\frac{2\kappa}{\mu^2}\widetilde{r}_b}\right)\\
		&\times\left[Q_1\!\left(
		\sqrt{\frac{\mu^2 \widetilde{\gamma}_{\text{th}} \widetilde{r}_b}{(1-\mu^2)(\widetilde{\gamma}_{\text{th}}+1)}},
		\sqrt{\frac{\mu^2 r_b}{(1-\mu^2)(\widetilde{\gamma}_{\text{th}}+1)}}
		\right) 
		- \frac{e^{	-\frac{\mu^2(\widetilde{\gamma}_{\text{th}} \widetilde{r}_b + r_b)}{2(1-\mu^2)(\widetilde{\gamma}_{\text{th}}+1)}}}{\widetilde{\gamma}_{\text{th}}+1}
		I_0\!\left(
		\frac{\mu^2 \sqrt{\widetilde{\gamma}_{\text{th}} r_b \widetilde{r}_b}}{(1-\mu^2)(\widetilde{\gamma}_{\text{th}}+1)}
		\right)\right]^{L_b}dr_bd\widetilde{r}_bd\phi_{x_1}d\phi_{y_1}\\
	\end{aligned}
\end{equation}
	\hrulefill
\end{figure*}
where the $\widetilde{\gamma}_{\text{th}}=\left[\frac{\left(\phi_{x_1}-\frac{D_1}{2}\right)^2+\left(\phi_{y_1}+\frac{D_2}{4}\right)^2+h^2}{\left(\phi_{x_1}-\frac{D_1}{2}\right)^2+\left(\phi_{y_1}-\frac{D_2}{4}\right)^2+h^2}\right]^{\frac{\varepsilon}{2}}\gamma_{\text{th}}$.
\end{corollary}
\begin{IEEEproof}
Based on \eqref{OP-MU-FA-PA}, once the antenna locations at the BS are fixed, the desired result can be readily obtained, which completes the proof.
\end{IEEEproof}
\section{Numerical Results}
In this section, the performance of the proposed HPFAS is evaluated through both theoretical analysis and Monte Carlo simulations. The carrier frequency is set to $f_c = 28$ GHz\footnote{This frequency is adopted as a representative mmWave setting to highlight spatial propagation characteristics. The proposed framework, however, is not restricted to mmWave bands and is applicable to other frequency ranges as well.}, with an effective refractive index of $n_{\text{eff}} = 1.4$ \cite{10945421}. {\color{black}The path-loss exponent is set to $\varepsilon = 2.5$. The noise power is set to $\sigma^2 = -80$ dBm.} The size of the FA is set to $W = 2$ with $N = 20$ ports, {\color{black}and the correlation parameter is set to $\mu^2 = 0.97$~\cite{10623405}.}
\begin{figure*}[t]
	\centering
	\begin{minipage}{0.3\textwidth}
		\centering
		\includegraphics[width=\linewidth]{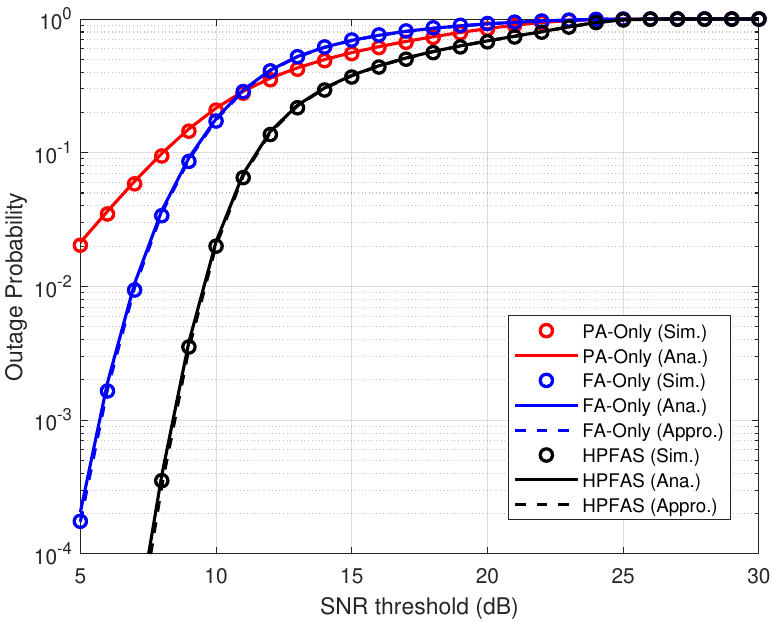}
		\caption{Outage probability versus the $\gamma_{\text{th}}$ in the single-user scenario with $h=3$~m, $D_1=20$~m, $D_2=20$~m, $P_t=15$~dBm, and $\kappa=7$.}
		\label{fig:OP-SNR-SU}
	\end{minipage}
	\hfill
	\begin{minipage}{0.3\textwidth}
		\centering
		\includegraphics[width=\linewidth]{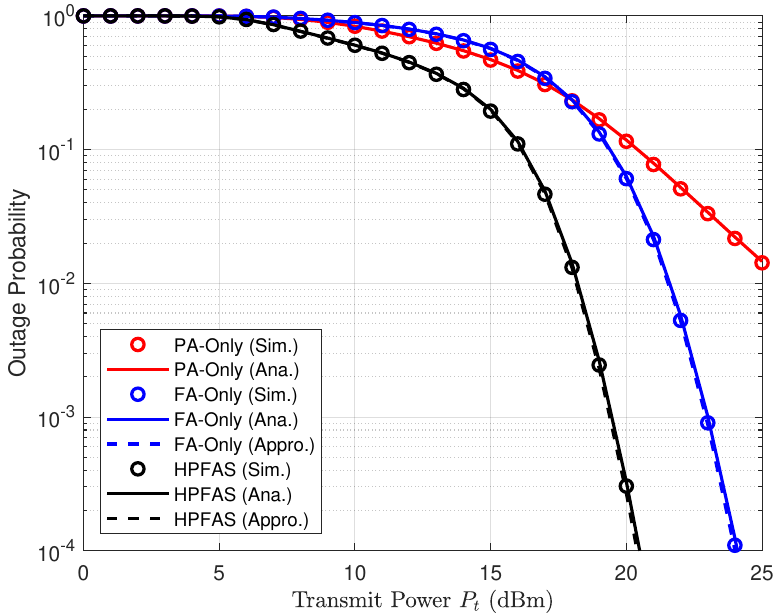}
		\caption{Outage probability versus the $P_t$ in the single-user scenario with $h=4$~m, $D_1=30$~m, $D_2=25$~m,  $\gamma_{\text{th}}=10$ dB, and $\kappa=5$.}
		\label{fig:OP-Pt-SU}
	\end{minipage}
	\hfill
	\begin{minipage}{0.3\textwidth}
		\centering
		\includegraphics[width=\linewidth]{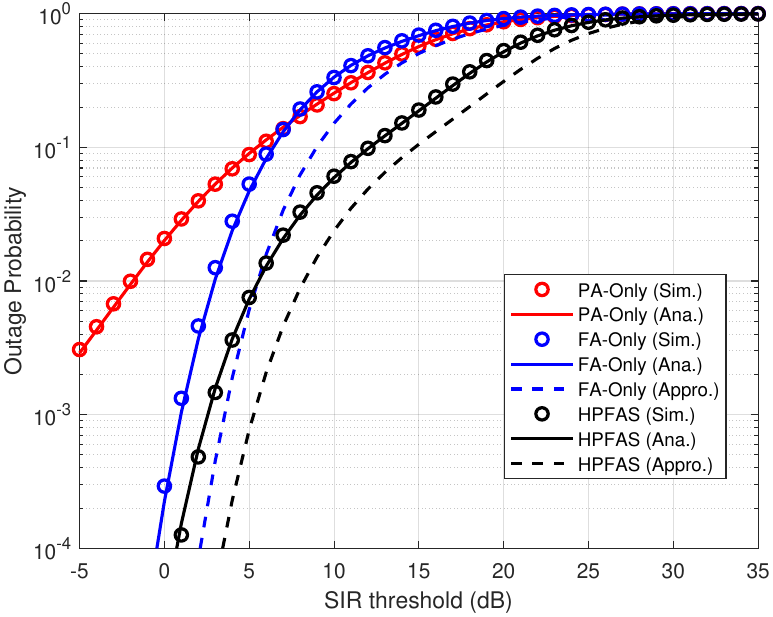}
		\caption{Outage probability versus the SIR threshold in the multi-user scenario with $h=3$, $D_1=D_2=25$ m, and $\kappa=5$.}
		\label{fig:OP-SNR-MU}
	\end{minipage}
\end{figure*}

\Cref{fig:OP-SNR-SU} illustrates the outage probability versus the SNR threshold $\gamma_{\text{th}}$ in the single-user scenario. The analytical results match the Monte Carlo simulations, validating the derived expressions.  The proposed HPFAS consistently outperforms the PA-only and FA-only schemes across all SNR thresholds due to the joint LoS enhancement and spatial diversity. When $D_1 = D_2 = 20$~m, the PA-only scheme performs better at high SNR thresholds where path loss dominates, while the FA-only scheme is superior at low SNR thresholds due to its diversity gain against deep fading. The approximated results closely match the exact analysis while significantly reducing complexity.

\Cref{fig:OP-Pt-SU} depicts the outage probability versus the transmit power $P_t$ for the single-user case. The outage probability decreases monotonically with increasing $P_t$, and the analytical results agree well with simulations, validating the theoretical analysis. The HPFAS scheme consistently outperforms the PA-only and FA-only baselines. 
{\color{black}At low $P_t$, the performance is limited by large-scale path loss, where the PA provides a direct gain by reducing the effective propagation distance. As $P_t$ increases, the system becomes fading-limited, where small-scale variations dominate and the FA becomes more effective due to its spatial diversity gain. The approximation closely matches the analytical results.}

\Cref{fig:OP-SNR-MU} shows the outage probability of a representative user in the two-waveguide, two-antenna, two-user scenario versus the SIR threshold $\gamma_{\text{th}}$. The outage probability increases monotonically with $\gamma_{\text{th}}$, as expected. The analytical results agree well with simulations, validating the analysis. In this interference-limited setting, HPFAS consistently outperforms the PA-only and FA-only baselines, demonstrating the benefit of jointly exploiting pinching-enhanced LoS and fluid-antenna diversity. 
{\color{black}The approximation becomes more accurate in the high-SIR regime, while its accuracy degrades at low SIR due to error accumulation from two successive SFA operations.}

\begingroup
\color{red}

\begin{figure*}[t]
	\centering
	\begin{minipage}{0.3\textwidth}
		\centering
		\includegraphics[width=\linewidth]
		{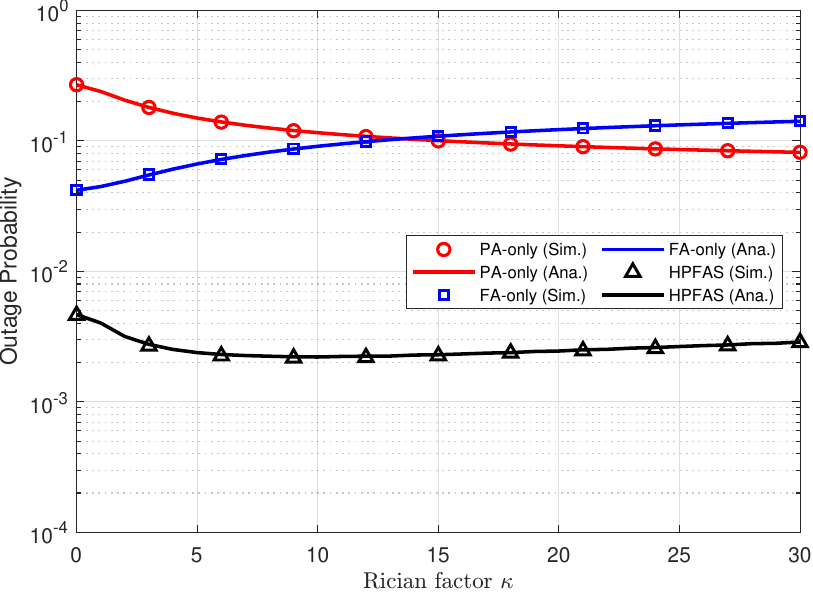}
		\caption{Outage probability versus $\kappa$ for
			$N=20$, $D_1=D_2=25$~m, $P_t=22.5$~dBm, and
			$\gamma_{\mathrm{th}}=14$~dB.}
		\label{fig:c3_kappa}
	\end{minipage}
	\hfill
	\begin{minipage}{0.3\textwidth}
		\centering
		\includegraphics[width=\linewidth]
		{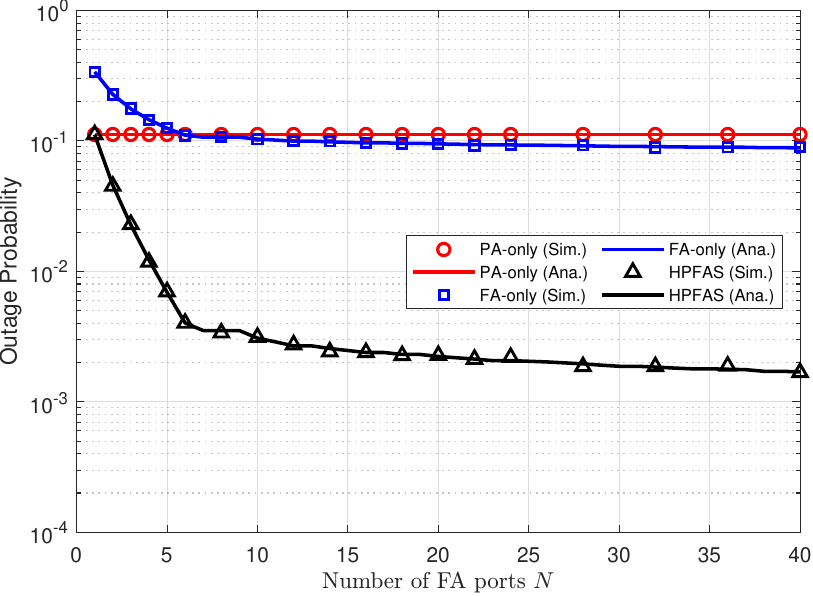}
		\caption{Outage probability versus $N$ for
			$\kappa=11$ and $D_1=D_2=25$~m.}
		\label{fig:c3_N}
	\end{minipage}
	\hfill
	\begin{minipage}{0.3\textwidth}
		\centering
		\includegraphics[width=\linewidth]
		{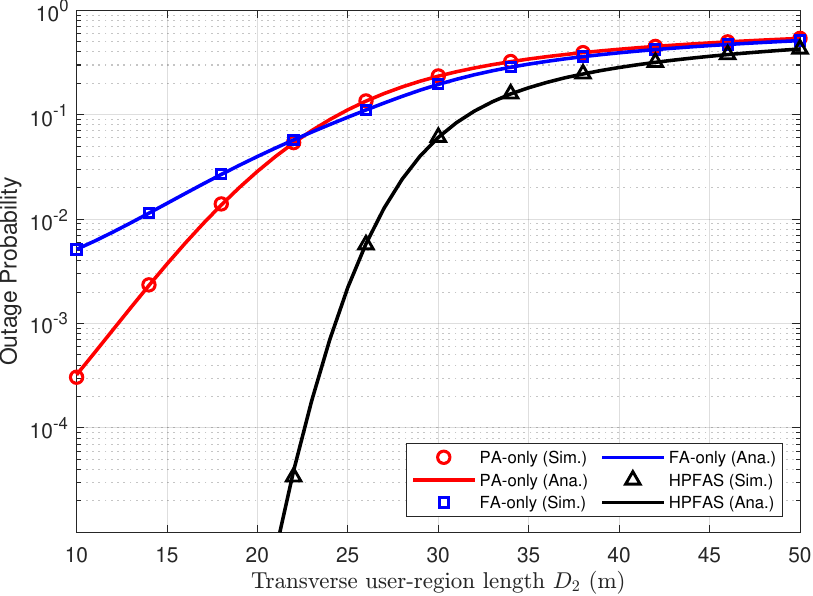}
		\caption{Outage probability versus $D_2$ for
			$\kappa=11$, $N=20$, and $D_1=25$~m.}
		\label{fig:c3_D2}
	\end{minipage}
\end{figure*}

\Cref{fig:c3_kappa} illustrates the outage probability versus the Rician
factor $\kappa$. The analytical results agree well with the Monte Carlo
simulations. The PA-only and FA-only curves intersect at
$\kappa^{*}\approx13.55$. FA-only performs better when
$\kappa<\kappa^{*}$ because the stronger scattered component provides
more effective port-selection diversity. In contrast, PA-only performs
better when $\kappa>\kappa^{*}$, where the stronger LoS component
increases the benefit of PA-enabled link-geometry adjustment. HPFAS
consistently achieves the lowest outage probability by jointly exploiting
both mechanisms.

\Cref{fig:c3_N} depicts the outage probability versus the number of FA
ports $N$. The analytical results closely match the simulation results.
Since PA-only does not employ FA port selection, its outage probability
is independent of $N$, whereas the performance of FA-only and HPFAS
improves as $N$ increases. The PA-only and FA-only curves intersect at
$N^{*}\approx5.93$. Therefore, PA-only performs better for
integer-valued $N\leq5$, while FA-only performs better for $N\geq6$
because of the increasing port-selection diversity. HPFAS maintains the
best outage performance over the considered range.

\Cref{fig:c3_D2} shows the outage probability versus the transverse
user-region length $D_2$. The outage probability increases with $D_2$
because of the increased propagation distance. The PA-only and FA-only
curves intersect at $D_2^{*}\approx22.61$~m. PA-only performs better
when $D_2<D_2^{*}$, where PA-enabled link-geometry adjustment provides
the dominant gain. As $D_2$ increases, the transverse distance becomes
more significant and reduces the relative benefit of PA adjustment.
Consequently, FA-only performs better when $D_2>D_2^{*}$, while HPFAS
consistently achieves the lowest outage probability.

\endgroup

\begingroup
\color{red}

\section{Conclusion}

This letter investigates HPFAS and analyzes its outage performance in
both single-user and multi-user interference scenarios. By modeling
PA-enabled link-geometry adjustment and FA-assisted spatial diversity,
exact integral-form outage probability expressions are derived and
validated via Monte Carlo simulations. High-SNR asymptotic analysis is
further conducted to characterize the diversity behavior of the
considered schemes. Numerical results show that the proposed HPFAS
consistently outperforms the PA-only and FA-only schemes over a wide
range of system parameters and maintains significant gains in the
interference-limited multi-user scenario. The comparisons among them are also provided to clarify the roles of the two spatial
reconfiguration mechanisms and offer scheme-selection insights.
Moreover, the SFA method significantly reduces computational complexity
while maintaining satisfactory approximation accuracy, particularly in
the single-user case.

\endgroup

 \bibliographystyle{IEEEtran}
 \bibliography{Reference}
\end{document}